\documentclass[runningheads]{llncs}
\usepackage[T1]{fontenc}
\usepackage{graphicx,xcolor}

\newcommand{\PBWT}{\ensuremath{\mathrm{PBWT}}}

\begin{document}

\title{Faster PBWT prefix-array access via batching}
\author{Travis Gagie\inst{1}\orcidID{0000-0003-3689-327X}}
\authorrunning{T. Gagie}
\institute{Center for Biotechnology and Bioengineering (CeBiB), Chile \email{travis.gagie@gmail.com}}
\maketitle

\begin{abstract}
\noindent
The positional Burrows-Wheeler Transform (PBWT) is commonly used to store haplotype panels compactly in such a way that, given a query haplotype, we can quickly find the set-maximal exact matches (SMEMs) between the query and the haplotypes in a panel.  There are generally two steps in this process: first we find the maximal substrings of the query that occur in the same positions in haplotypes in the panel and then, for each such substring, report the IDs of the haplotypes in the panel in which the substring occurs in the same position as in the query.  Very recently, Bonizzoni, Gagie and Gao (2026) gave two time-space tradeoffs for the second step: they use either $O ((r + h) \log n)$ bits and $O (\log \log_w \min (h, \ell) + k)$ time to report $k$ haplotype IDs, or $O (r \log h + h \log n)$ bits and $O (k \log \log_w h)$ time, where $r$ is the number of runs in the panel's PBWT, $h$, $\ell$ and $n = h \ell$ are the panel's height, length and size, respectively, and $w$ is the number of bits in a machine word.  We observe here that if we can batch queries until we have found $r \lg (h) / \lg r$ such substrings, then for the second step we can easily use $O (r \log h)$ bits and $O (r + K)$ total time to report all $K$ haplotype IDs, which is $O (\log (r) / \log h)$ amortized time for each haplotype ID, assuming the haplotypes in the panel are distinct.  With that assumption, this is the same space bound as for Bonizzoni, Gagie and Gao's second tradeoff, with a time bound that is asymptotically smaller when $\lg (r) / \lg h \in o (\log \log_w h)$.  Compared to their first tradeoff, our space bound is smaller by a factor of about $\lg (n) / \lg h$ and our time bound is larger by a factor of at most $\lg (r) / \lg h$.  Our approach is based on an algorithm for constructing the prefix arrays quickly from the PBWT, which may be of independent interest.

\end{abstract}

\section{Introduction}
\label{sec:introduction}

Human beings are diploid, meaning that our cells generally contain two copies of each chromosome in the genome.  Each copy of a chromosome can be expressed as the characters it has at the variation sites in that chromosome, and these characters together are called a {\em haplotype} (for ``haploid genotype'').  Because the human genome is so big, the probability of different mutations occurring at the same site and becoming widespread in the population is small, so haplotypes are generally expressed as binary strings and collections of haplotypes are expressed as binary matrices, called {\em haplotype panels}, with a row for each haplotype and a column for each variation site; a haplotype's ID is the number of its row.  All the results in this paper hold, however, as long as there are only a constant number of variations at each site.

Geneticists are interested in storing haplotype panels such that, given a query haplotype, they can quickly find the maximal substrings of the query that occur in the same positions in haplotypes in the panel and, for each such substring, report the IDs of the haplotypes in the panel in which the substring occurs in the same position as in the query.  We call finding the substrings the {\em first step} and reporting the haplotype IDs the {\em second step}.  Such a substring in the query and a {\em match} for it --- that is, an occurrence of it in the same position in a haplotype in the panel --- are together called a {\em set-maximal exact match} (SMEM), and finding SMEMs is useful for inferring how the person from whom the query haplotype came is likely related to the people from whom the haplotypes in the panel came, for example, or filling in missing values in the query haplotype; see~\cite{ICGV+26} for a list of applications.  Because the number of variation sites in the human genome is fairly large, storing haplotype panels compactly is also useful.

Inspired by the Burrows-Wheeler Transform~\cite{BW94} (BWT) and Ferragina and Manzini's~\cite{FM05} FM-index based on it, Durbin~\cite{Dur14} proposed the {\em positional BWT} (PBWT) on haplotype panels.  To build the PBWT of a binary matrix $M [0..h - 1][0..\ell - 1]$ we stably sort the bits in each column $M [0..h - 1][j]$ into the co-lexicographic order of the prefixes $M [0..h - 1][0..j - 1]$ of the rows preceding those bits.  Haplotypes are more likely to agree at a variation site if they have agreed at the preceding variation sites, so the columns of PBWTs usually consist of relatively few unary runs.  The {\em prefix array} $a_j$ for column $j$ of the PBWT is the permutation of $\{0, \ldots, h - 1\}$ such that $a_j [i]$ indicates the row from which $\PBWT [i][j]$ came.  Conceptually, we append a dummy bit to each haplotype to define an extra prefix array $a_\ell$.  Figure~\ref{fig:example2} shows two copies of a toy haplotype panel with $h = \ell = 8$ (from~\cite{ICGV+26}), its PBWT and the matrix whose columns are the prefix arrays; we explain the highlighting below.

Let $n = h \ell$ be the size of the panel and let $r$ be the total number of runs in the columns of its PBWT; in practice, $r \ll n$.  If we concatenate the columns of the PBWT and store a sparse $n$-bit bitvector with 1s marking the beginnings of the runs and the same for the PBWT of the matrix with the rows reversed, then we use $O (r \log (n / r) + r' \log (n / r') + \ell \log n)$ bits, where $r'$ is the number of runs in the columns of the PBWT of the reversed matrix.  By using a version of Li's~\cite{Li12} forward-backward algorithm, for example, we can then find the substrings of a query haplotype that are in SMEMs in time proportional to the total length of those substrings times $O (\log \log n)$; see~\cite{ICGV+26} for more details.  Moreover, for each such substring in the query we find
\begin{itemize}
\item the interval in the prefix array for the PBWT column immediately after the substring's matches, that contains the IDs of the haplotypes containing those matches,
\item the last entry in that interval of the prefix array.
\end{itemize}
Recently, Bonizzoni, Cozzi and Gao~\cite{BCG26} (see also~\cite[Section 4]{BGR21}) showed how to store the run-length compressed PBWT in $O (r)$ words such that we can step forward and backward in the PBWT repeatedly in constant time per step.  Their result lets us run the forward-backward algorithm with constant-time steps, using a total of $O (r + r')$ words for the PBWT and the PBWT of the matrix with the rows reversed, and thus find those substrings in the query, the lengths of their intervals in the prefix arrays and those intervals' last entries, all in time proportional to the substrings' total length once we have read the query, removing the $O (\log \log n)$ factor in the time bound.  It is useful to find the intervals' last entries because the $\phi$ function~\cite{KMP09} has been adapted from the suffix array to the PBWT~\cite{CRR+23} and we can use it to start from the last entries and enumerate all the entries.

Figure~\ref{fig:example2} also shows the query haplotypes $Q_1 = \mathtt{10100111}$ and $Q_2 = \mathtt{11100010}$ and the highlighted matches in the panel for their substrings $Q_1 [0..3], Q_1 [2..6],$ $Q_1 [3..7]$ and $Q_2 [0..1], Q_2 [1..4],$ $Q_2 [5..7]$ that are in SMEMs.  The corresponding bits in the PBWT and the entries in the prefix arrays are also highlighted.  Notice the highlighting for each match continues one prefix array to the right after the match ends, because we can only be sure to find an interval in that prefix array containing the IDs of the haplotypes we should report (that is, $a_4 [2], a_7 [7], a_8 [7]$ for $Q_1$ and $a_2 [6..7], a_5 [1], a_8 [1..2]$ for $Q_2$).  For these examples there are such intervals in the earlier prefix arrays as well, but that is just coincidence.

\begin{figure}[t]
\begin{center}
\begin{tabular}{c@{\hspace{5ex}}c}
\begin{tabular}{r|ccc}
  & haplotype    &              & prefix        \\
  & panel:       & PBWT:        & arrays:       \\[.5ex]
  & \tt 01234567 & \tt 01234567 & \tt 012345678 \\
\hline
0 & \tt 00110100 & \tt 00111\textcolor{green}{\bf 1}01 & \tt 00041\textcolor{green}{\bf 6}440 \\
1 & \tt 01001101 & \tt 0100\textcolor{green}{\bf 0}\textcolor{blue}{\bf 1}10 & \tt 1141\textcolor{green}{\bf 6}\textcolor{blue}{\bf 3}202 \\
2 & \tt 01011010 & \tt 01\textcolor{red}{\bf 1}11111 & \tt 22\textcolor{red}{\bf 5}2\textcolor{red}{\bf 5}0717 \\
3 & \tt 01\textcolor{blue}{\bf 1}\textcolor{blue}{\bf 0}\textcolor{blue}{\bf 0}\textcolor{blue}{\bf 1}\textcolor{blue}{\bf 1}0 & \tt 010\textcolor{green}{\bf 0}\textcolor{blue}{\bf 0}1\textcolor{green}{\bf 1}1 & \tt 331\textcolor{green}{\bf 6}\textcolor{blue}{\bf 3}1\textcolor{green}{\bf 6}53 \\
4 & \tt 10011001 & \tt 100111\textcolor{blue}{\bf 1}0 & \tt 442745\textcolor{blue}{\bf 3}24 \\
5 & \tt \textcolor{red}{\bf 1}\textcolor{red}{\bf 0}\textcolor{red}{\bf 1}\textcolor{red}{\bf 0}1101 & \tt \textcolor{red}{\bf 1}\textcolor{red}{\bf 0}\textcolor{blue}{\bf 1}11000 & \tt \textcolor{red}{\bf 5}\textcolor{red}{\bf 5}\textcolor{blue}{\bf 3}024071 \\
6 & \tt 110\textcolor{green}{\bf 0}\textcolor{green}{\bf 0}\textcolor{green}{\bf 1}\textcolor{green}{\bf 1}\textcolor{green}{\bf 1} & \tt 110\textcolor{red}{\bf 0}100\textcolor{green}{\bf 1} & \tt 666\textcolor{red}{\bf 5}721\textcolor{green}{\bf 6}5 \\
7 & \tt 11011010 & \tt 110\textcolor{blue}{\bf 0}0000 & \tt 777\textcolor{blue}{\bf 3}075\textcolor{blue}{\bf 3}\textcolor{green}{\bf 6} \\
\hline
$Q_1$ & \tt 10100111
\end{tabular}
&
\begin{tabular}{r|ccc}
  & haplotype    &              & prefix        \\
  & panel:       & PBWT:        & arrays:       \\[.5ex]
  & \tt 01234567 & \tt 01234567 & \tt 012345678 \\
\hline
0 & \tt 00110100 & \tt 00111101 & \tt 000416440 \\
1 & \tt 01001101 & \tt 010001\textcolor{orange}{\bf 1}0 & \tt 11416\textcolor{green}{\bf 3}\textcolor{orange}{\bf 2}0\textcolor{orange}{\bf 2} \\
2 & \tt 01011\textcolor{orange}{\bf 010} & \tt 011111\textcolor{pink}{\bf 1}1 & \tt 225250\textcolor{pink}{\bf 7}1\textcolor{pink}{\bf 7} \\
3 & \tt 0\textcolor{green}{\bf 1100}110 & \tt 0\textcolor{green}{\bf 1}00\textcolor{green}{\bf 0}111 & \tt 3\textcolor{green}{\bf 3}16\textcolor{green}{\bf 3}1653 \\
4 & \tt 10011001 & \tt 1001111\textcolor{orange}{\bf 0} & \tt 4427453\textcolor{orange}{\bf 2}4 \\
5 & \tt 10101101 & \tt 10\textcolor{green}{\bf 1}1100\textcolor{pink}{\bf 0} & \tt 55\textcolor{green}{\bf 3}0240\textcolor{pink}{\bf 7}1 \\
6 & \tt \textcolor{red}{\bf 11}000111 & \tt \textcolor{red}{\bf 11}001\textcolor{orange}{\bf 0}01 & \tt \textcolor{red}{\bf 666}57\textcolor{orange}{\bf 2}165 \\
7 & \tt \textcolor{blue}{\bf 11}011\textcolor{pink}{\bf 010} & \tt \textcolor{blue}{\bf 11}0\textcolor{green}{\bf 0}0\textcolor{pink}{\bf 0}00 & \tt \textcolor{blue}{\bf 777}\textcolor{green}{\bf 3}0\textcolor{pink}{\bf 7}536 \\
\hline
$Q_2$ & \tt 11100010
\end{tabular}
\end{tabular}
\caption{Two copies of a toy haplotype panel with $h = \ell = 8$, its PBWT and the matrix whose columns are the prefix arrays.  The highlighting shows the matches in the SMEMs of $Q_1 = \mathtt{10100111}$ {\bf (left copy)} and $Q_2 = \mathtt{11100010}$ {\bf (right copy)}, the corresponding bits in the PBWT and the corresponding entries in the prefix arrays.}
\label{fig:example2}
\end{center}
\end{figure}

Even more recently, Bonizzoni, Gagie and Gao~\cite{BGG26} showed how to implement the $\phi$ function such that we store $O ((r + h) \log n)$ bits and report the contents of each of those intervals in the prefix arrays in $O (\log \log_w \min (h, \ell) + k)$ time, where $k$ is the length of the interval and $w \in \Theta (\log n)$ is the number of bits in a machine word in the word-RAM model, or store $O (r \log h + h \log n)$ bits and report the contents of each interval in $O (k \log \log_w h)$ time.  This raises an obvious question: when can we store $O (r \log h + h \log n)$ bits and report the contents of each interval in $O (\log \log_w \min (h, \ell) + k)$ time or, even better, in $O (k)$ time?  We observe in Section~\ref{sec:batching} that if we can batch queries until we have found $r \lg (h) / \lg r$ substrings in the queries that are in SMEMs, then we can easily use $O (r \log h)$ bits and $O (r + K)$ total time to report all $K$ haplotype IDs in all those substrings' intervals, which is $O (\log (r) / \log h)$ amortized time for each haplotype ID, assuming the haplotypes in the panel are distinct.  This only partly answers our question: our time bound is amortized over the batch, and it is $O (k)$ for an interval of length $k$ only when $r$ is polynomial in $h$.  We compare our bounds to Bonizzoni, Gagie and Gao's tradeoffs at the end of Section~\ref{sec:batching}.  Our approach is based on an algorithm for constructing the prefix arrays quickly from the PBWT, which may be of independent interest.  In the appendix we give a similar result for the divergence arrays.

\section{Batching and reporting}
\label{sec:batching}

Suppose we have the PBWT for a panel with the columns run-length compressed and, for each run in a column, the last entry in the corresponding interval in that column's prefix array, as in the $\mu$-PBWT~\cite{CRR+23}, for example.  This takes a total of $O (r \log h)$ bits.  Moreover, suppose that for $q$ query haplotypes we have already found the substrings that are in SMEMs with respect to the panel.  Finally, suppose that for each such substring, we have already found the length of the interval in the appropriate prefix array (for the column right after where the substring ends) that contains the IDs of the haplotypes containing the substring's matches, and the last entry in that interval.  We assume we find each query's substrings in increasing order by their ending positions.

For each such substring we store a quintuple: which query the substring belongs to, where the substring starts and ends, the length of the substring's interval in the appropriate prefix array and the last entry in that interval.  We need not store where the interval starts and ends because, as we explain below, we reach the interval through its last entry.  We do not store the queries themselves.  For our example from Figure~\ref{fig:example2}, the quintuples sorted by the substrings' ending positions are
\[(2, 0, 1, 2, 7),\ (1, 0, 3, 1, 5),\ (2, 1, 4, 1, 3),\ (1, 2, 6, 1, 3),\ (1, 3, 7, 1, 6),\ (2, 5, 7, 2, 7)\,.\]

Our idea is to process queries and to stop after the first query with which we have found at least $r \lg (h) / \lg r$ such substrings.  We assume $h \leq r$, which holds whenever the haplotypes in the panel are distinct~\cite[Corollary 5]{BCG26}; we discuss panels with duplicate haplotypes after Theorem~\ref{thm:batching}.  We also have $\ell \leq r$, because every column of the PBWT contains at least one run.  We do not count queries that contain no such substrings, so each of the $q$ queries we process contains at least one substring and $q \leq r \lg (h) / \lg r + 1 \leq r + 1$.  Since also $\ell \leq r$ and $h \leq r$, each quintuple takes $O (\log q + \log \ell + \log h) \subseteq O (\log r)$ bits, so we use
\[O \left( \frac{r \log h}{\log r} \cdot \log r \right)
= O (r \log h)\]
bits to store quintuples for the first $\lceil r \lg (h) / \lg r \rceil$ substrings we find.  The query containing the last of those substrings is the last query we process.  The last query can contain as many as $\ell - 1$ more substrings, so storing quintuples for those remaining substrings could take $\Theta (\ell \log r)$ bits, which is not necessarily in $O (r \log h)$.  Those remaining substrings all belong to the same query, however, and since they are sorted by their ending positions they are also sorted by their starting positions, because substrings for the same query cannot nest.  We therefore store those remaining substrings in a separate list, without recording the query and with two $\ell$-bit bitvectors marking where those remaining substrings start and end, which takes $O (\ell \log h) \subseteq O (r \log h)$ bits and $O (\ell) \subseteq O (r)$ time on top of the time to find those remaining substrings.

When we stop, the separate list is already sorted by the substrings' ending positions.  We sort the quintuples by the substrings' ending positions with a radix sort with radix $\lceil \sqrt{r} \rceil$, which takes two passes because the ending positions are less than $\ell \leq r$, and thus $O (r)$ time and $O (r \log h)$ bits including $O (\sqrt{r} \log r)$ bits for the counters.  (A single pass with radix $\ell$ could take $\Theta (\ell \log r)$ bits for the counters.)

Let $K$ be the total number of haplotype IDs we will report --- counting repetitions for different queries.  We generate all the prefix arrays one after the other according to the lemma below, which takes $O (h + r) = O (r)$ time and $O (h \log h) \subseteq O (r \log h)$ bits, and extract the haplotype IDs to report them as we go.  Since each substring has at least one match in the panel, $K \geq r \lg (h) / \lg r$ and so $O (r) \subseteq O (K \log (r) / \log h)$.

\begin{lemma}
\label{lem:building_PAs}
Suppose we already have the PBWT with run-length compressed columns and, for each run in a column, the last entry in the corresponding interval in that column's prefix array.  Then we can generate all the prefix arrays one after the other, in a doubly-linked list, in $O (h + r)$ total time and $O (h \log h)$ bits of working space.
\end{lemma}

\begin{proof}
We start by building a doubly-linked list containing the numbers from 0 to $h - 1$ --- that is, the contents of prefix array $a_0$ --- and a static array $A [0..h - 1]$ of pointers in which $A [x]$ points to the node of the doubly-linked list storing $x$.  We maintain the invariant that, after we have processed $j$ columns of the PBWT, the doubly-linked list contains $a_j$ and we know $a_j [0]$.

For each column $\PBWT [0..h - 1][j]$ of the PBWT in turn and each position $i < h - 1$ such that $\PBWT [i][j] \neq \PBWT [i + 1][j]$ --- so $a_j [i]$ is the last entry in the interval for a run and thus stored --- we follow the pointer $A [a_j [i]]$ to the node of the doubly-linked list storing $a_j [i]$, follow that node's pointer to its successor node in the doubly-linked list, and then cut the list between them.  When we are finished, we have cut the doubly-linked list into pieces corresponding to the runs in $\PBWT [0..h - 1][j]$ and we have pointers to the nodes on either side of each cut.  The first node of the first piece is the node storing $a_j [0]$, which we know by our invariant, and the last node of the last piece is the node storing $a_j [h - 1]$, which is stored because $a_j [h - 1]$ is the last entry in the interval for the last run; thus, we have pointers to the first and last nodes of every piece.  We stably sort those pieces according to the bits in the corresponding runs and then rejoin them, noting which value appears first in the first piece.  This takes time proportional to the number of runs in $\PBWT [0..h - 1][j]$ and, when we are done, the doubly-linked list contains  $a_{j + 1}$ and we know $a_{j + 1} [0]$.  Processing all the columns of the PBWT and generating the prefix arrays one after the other takes $O (h + r)$ time and $O (h \log h)$ bits on top of the PBWT.  In our example, when we start processing column 3 of the PBWT, the doubly-linked list contains $a_3 = [4, 1, 2, 6, 7, 0, 5, 3]$.  The runs in $\PBWT [0..7][3] = \mathtt{10101100}$ start at positions 0, 1, 2, 3, 4 and 6, so we cut the doubly-linked list into pieces $(4), (1), (2), (6), (7, 0), (5, 3)$ and reassemble them into $a_4 = [1, 6, 5, 3, 4, 2, 7, 0]$.
\end{proof}

\noindent
Lemma~\ref{lem:building_PAs} can be viewed as a version of Durbin's~\cite{Dur14} Algorithm 1 (``BuildPrefixArray'') that takes advantage of already having the PBWT.  Sanaullah, Zhi and Zhang's~\cite{SZZ21} d-PBWT also stores prefix arrays in linked lists, but the d-PBWT stores a linked list for every column, in order to support inserting and deleting haplotypes, whereas we use a single linked list with $h$ nodes.

Since we have the substrings' quintuples sorted by the columns where the substrings end --- in two lists, which we merge as we go --- when the doubly-linked list contains the prefix array for the column immediately after the one where one of the substrings ends, we can report the haplotype IDs in that substring's interval in that prefix array in time proportional to the interval's length: we follow the pointer in the array $A$ from the proof of Lemma~\ref{lem:building_PAs} to the node of the doubly-linked list storing the last entry in the interval and then walk backward through the doubly-linked list, reporting as many haplotype IDs as the length of the interval.  We thus report the haplotype IDs in order by the columns where the substrings end, each with its substring's query and starting and ending positions, and not in order by query.  It follows that we can report all $K$ haplotype IDs for all the SMEMs in $O (K)$ time, in addition to the $O (r)$ time we spend building the prefix arrays.  This gives us the following theorem.

\begin{theorem}
\label{thm:batching}
Suppose $2 \leq h \leq r$ and we have the PBWT for a haplotype panel with the columns run-length compressed and, for each run in a column, the last entry in the corresponding interval in that column's prefix array.  Suppose also that we are allowed to batch queries to the panel, stopping after the first query with which we have found at least $r \lg (h) / \lg r$ substrings of the query haplotypes that are in SMEMs, and that we are given each query's substrings in increasing order by their ending positions and, for each of those substrings, where that substring starts and ends, the length of that substring's interval in the prefix array for the column immediately after the one where that substring ends, and the last entry in that interval.  Then on a word RAM we can report, when the batch is complete and for each of those substrings, the IDs of the haplotypes containing that substring's matches, using a total of $O (r \log h)$ bits including the PBWT and $O (r + K) \subseteq O (K \log (r) / \log h)$ total time, where $K$ is the total number of haplotype IDs we report.
\end{theorem}

\noindent
Theorem~\ref{thm:batching} covers only the second step.  For the first step, we can use any algorithm that finds the substrings, the lengths of their intervals and those intervals' last entries, such as the algorithms we discussed in Section~\ref{sec:introduction}.

By Theorem~\ref{thm:batching}, for the second step we use $O (\log (r) / \log h)$ amortized time for each haplotype ID.  Since $r \leq n = h \ell$, we have $\lg (r) / \lg h \leq 1 + \lg (\ell) / \lg h$.  For example, $\lg (r) / \lg h \leq 3$ whenever $h \geq 1000$ and $r \leq 10^9$, and $\lg (r) / \lg h \leq 4.5$ whenever $h \geq 100$ and $r \leq 10^9$.  The batches are correspondingly large, however: for $h = 1000$ and $r = 10^9$, we batch about $r / 3$ substrings.  If we run out of queries before we have found $r \lg (h) / \lg r$ substrings, then our bounds on the space and the total time still hold but our bound on the amortized time does not.

Since $2 \leq h \leq r$ and $\ell \leq r$, we have $h \lg \ell \leq h \lg r \in O (r \log h)$ --- because $x / \lg x$ is increasing for $x \geq 3$ --- and thus $h \lg n = h \lg h + h \lg \ell \in O (r \log h)$, so Bonizzoni, Gagie and Gao's second tradeoff also uses $O (r \log h)$ bits, with $O (\log \log_w h)$ time for each haplotype ID.  Batching is asymptotically faster when $\lg (r) / \lg h \in o (\log \log_w h)$ --- for example, when $r$ is polynomial in $h$ --- at the cost of reporting the haplotype IDs only when the batch is complete.  That advantage is only asymptotic, however: for $h = 1000$, $\ell = 10^6$ and $r = 10^9$, for example, we have $\lg (r) / \lg h = 3$ and $\lg \log_w h \approx 1$ with $w = \lg n$.

Their first tradeoff uses $O ((r + h) \log n)$ bits, so its space bound is larger than ours by a factor of about $\lg (n) / \lg h \leq 1 + \lg (r) / \lg h$.  Their first tradeoff uses $O (\log \log_w \min (h, \ell) + k)$ time for an interval of length $k$, for which we use $O (k \log (r) / \log h)$ amortized time.  When $r$ is polynomial in $h$ we have $\lg n \in O (\log h)$, because $n = h \ell \leq h r$, so their first tradeoff then also uses $O (r \log h)$ bits and batching uses constant amortized time for each haplotype ID; in that case, compared to their first tradeoff, batching only removes the $O (\log \log_w \min (h, \ell))$ term for each interval.

If the panel contains duplicate haplotypes, then we can apply Theorem~\ref{thm:batching} to the panel that contains one representative of each set of identical haplotypes, whose PBWT has at most as many runs as the PBWT of the original panel.  Suppose the haplotypes are numbered such that identical haplotypes are consecutive --- as they are when the haplotypes are sorted, for example --- and we store a sparse bitvector marking the first haplotype in each set of identical haplotypes, which takes $O (d \log (h / d) + d) \subseteq O (r \log h)$ bits, where $d \leq r$ is the number of distinct haplotypes in the original panel and $h$ and $r$ are the height of the original panel and the number of runs in the original panel's PBWT.  Then whenever we report a representative, we can instead report the IDs of all the haplotypes identical to that representative, in constant time for each haplotype ID.  If we must report haplotype IDs according to some other numbering, then we can store a permutation in $O (h \log h)$ bits, which is not necessarily in $O (r \log h)$.

Some researchers~\cite{BG26} are now adapting results on suffixient sets~\cite{DGLMP24}, suffixient arrays~\cite{CDGK+24} and suffix-tree path decompositions~\cite{BCGK+26} to haplotype panels.  Instead of producing each substring's interval in the appropriate prefix array, this will likely produce the first and last entries in that interval.  We note that our technique for reporting the haplotype IDs in the intervals can easily be adapted to work with this new information: after we build a prefix array, for each interval of haplotype IDs we should report from it, we follow the pointer in $A$ to the node of the doubly-linked list storing the interval's first entry and advance through the doubly-linked list --- reporting the IDs as we go --- until we reach the interval's last entry.

\begin{credits}

\subsubsection{\ackname} This research was funded by NSERC Discovery Grant RGPIN-07185-2020.  Many thanks to Paola Bonizzoni and Younan Gao for helpful discussions.

\subsubsection{\discintname} The author has no competing interests to declare that are relevant to the contents of this article.

\end{credits}


\appendix

\section{Building absolute divergence arrays}
\label{app:div_arrays}

For the sake of completeness, we now prove a slightly slower analogue of Lemma~\ref{lem:building_PAs} for encodings of the divergence arrays.  In the {\em divergence array} $d_j [0..h - 1]$ for column $j$ of the PBWT, we set $d_j [0] = 0$ by convention and the value $d_j [i]$ for $i > 0$ is the length of the longest common suffix of $M [a_j [i]][0..j - 1]$ and $M [a_j [i - 1]][0..j - 1]$.  Just as the prefix arrays are similar to the suffix array for a string, the divergence arrays are similar to the longest common prefix (LCP) array for a string.  For convenience, we define the {\em absolute divergence array} $d'_j$ such that $d'_j [i] = (j - 1) - d_j [i]$.  That is, for $i > 0$, $d'_j [i]$ is the previous column in which $M [a_j [i]][0..j - 1]$ and $M [a_j [i - 1]][0..j - 1]$ differ (or $-1$ if they are equal), rather than the distance from column $j - 1$ to that column.

\begin{lemma}
\label{lem:building_DAs}
Suppose we already have the PBWT with run-length compressed columns and, for each run in a column, the last entry in the corresponding interval in that column's prefix array.  Then we can generate all the absolute divergence arrays one after the other in $O (h + r \log h)$ total time.
\end{lemma}

\begin{proof}
All the entries in $d'_0$ are $-1$.  Building $d'_{j + 1}$ from $d'_j$ is almost like building $a_{j + 1}$ from $a_j$ but more complicated at the beginnings of runs in the PBWT.  To see why, consider that if $\PBWT [i - 1][j] = \PBWT [i][j]$ then the previous column in which $M [a_j [i]][0..j]$ and $M [a_j [i - 1]][0..j]$ differ is the same as the previous column in which $M [a_j [i]][0..j - 1]$ and $M [a_j [i - 1]][0..j - 1]$ differ, so we can get most of $d'_{j + 1}$ correct just by storing each value $d'_j [i]$ in the same node of the doubly-linked list as $a_j [i]$ and splitting and joining the doubly-linked list as we did to get $a_{j + 1}$ from $a_j$ in the proof of Lemma~\ref{lem:building_PAs}.

If $i = 0$ or $\PBWT [i - 1][j] \neq \PBWT [i][j]$, however --- meaning $d'_j [i]$ is at the beginning of one of the pieces into which we split the doubly-linked list --- then we must work harder.  Specifically, we let $i'$ be the last position before $i$ such that $\PBWT [i'][j] = \PBWT [i][j]$, so the ending position of the previous run of the same character.  (When the matrices are binary, $i'$ is the ending position not of the run before the one starting at $\PBWT [i][j]$ but of the run before that.)  If $i'$ does not exist, then we change the value at the beginning of that piece to $j$ (corresponding to a divergence value of 0).  Otherwise, we change the value at the beginning of that piece to the previous column in which $M [a_j [i]][0..j - 1]$ and $M [a_j [i']][0..j - 1]$ differ.

By the definition of the PBWT, that previous column is $\max (d'_j [i' + 1..i])$ --- the rightmost column on which any of the row-prefixes in $M [a_j [i'..i]] [0..j - 1]$ differ --- which we can find with a range-maximum query on $d'_j$.  To support that query, we keep a dynamic range-maximum data structure over the nodes of the doubly-linked list, allowing splits and joins corresponding to how we split and re-join the doubly-linked list.  For each column, we perform all the range-maximum queries before we change any values or cut the doubly-linked list, so that all the range-maximum queries are on $d'_j$.  If we implement the dynamic range-maximum data structure with an augmented AVL tree with $h$ nodes then range-maximum queries between two given nodes, updating single values, splitting and joining all take $O (\log h)$ time.  It follows that we can compute all the absolute divergence arrays one after the other in $O (h + r \log h)$ time.
\end{proof}

\noindent
Lemmas~\ref{lem:building_PAs} and~\ref{lem:building_DAs} together can be viewed as a version of Durbin's~\cite{Dur14} Algorithm 2 (``BuildPrefixAndDivergenceArrays'') that takes advantage of already having the PBWT.  It may also be possible to adapt Durbin's Algorithm 3 (``ReportLongMatches'') to find pairs of haplotypes in the panel that agree on long intervals of columns more quickly when we already have the PBWT.  This may be interesting when combined with Gagie's~\cite{Gag26} recent algorithm for merging PBWTs.

\end{document}